\documentclass[preprint]{amsart}

\usepackage{amsmath,amsthm,amssymb}

\theoremstyle{definition}

\newtheorem{thm}{Theorem}
\newtheorem{crl}{Corollary}

\DeclareMathOperator*{\esssup}{ess\,sup}

\begin{document}
\title{K\"ahler information manifolds of signal processing filters in weighted Hardy spaces}
\author{Jaehyung Choi}
\address{Goldman Sachs \& Co., 200 West Street, New York, NY 10282, USA}
\email{jj.jaehyung.choi@gmail.com}

\begin{abstract}
	We extend the framework of K\"ahler information manifolds for complex-valued signal processing filters by introducing weighted Hardy spaces and smooth transformations of transfer functions. We demonstrate that the Riemannian geometry induced from weighted Hardy norms for the smooth transformations of its transfer function is a K\"ahler manifold. In this setting, the K\"ahler potential of the linear system geometry corresponds to the squared weighted Hardy norm of the composite transfer function. With the inherent structure of K\"ahler manifolds, geometric quantities on the manifold of linear systems in weighted Hardy spaces can be computed more efficiently and elegantly. Moreover, this generalized framework unifies a variety of well-known information manifolds within the structure of K\"ahler information manifolds for signal filters. Several illustrative examples from time series models are provided, wherein the metric tensor, Levi-Civita connection, and K\"ahler potentials are explicitly expressed in terms of polylogarithmic functions of the poles and zeros of transfer functions parameterized by weight vectors.
\end{abstract}

\maketitle
\section{Introduction}
Information geometry has emerged as a vibrant interdisciplinary field, bridging concepts across diverse research areas. Notably, the differential geometric approaches have led to both theoretical advancements and practical enhancements in time series analysis and signal processing. Since the foundational studies on the information geometry of time series models such as autoregressive and moving average (ARMA) processes, as well as fractionally integrated ARMA (ARFIMA) processes were introduced \cite{amari1987differential, ravishanker1990differential, ravishanker2001differential}, the geometric structure of these models has been further explored through symplectic geometry and K\"ahler geometry \cite{barndorff1997statistics, barbaresco2006information, barbaresco2012information, zhang2013symplectic, barbaresco2014koszul, choi2015kahlerian}. Beyond the theoretical insights, information geometry has also yielded practical applications, particularly in the development of improved Bayesian predictive priors \cite{komaki2006shrinkage, tanaka2008superharmonic, tanaka2018superharmonic, choi2015kahlerian, choi2015geometric, oda2021shrinkage}. In this context, the shrinkage priors for time series models are easily formulated within the K\"ahler extension of information geometry for time series and signal filters \cite{choi2015kahlerian, choi2015geometric, oda2021shrinkage}.
	
	Information geometry, along with its K\"ahlerian formulation, in time series analysis and signal processing is intimately connected to the theory of Hardy spaces in functional analysis. A fundamental assumption in the information geometry of linear systems is that the logarithmic power spectrum belongs to a Hardy space \cite{amari1987differential, amari2000methods}. A similar assumption underlies the K\"ahler structure of information geometry for linear systems: the logarithmic transfer function of a signal filter has a finite Hardy norm \cite{choi2015kahlerian}. In K\"ahlerian linear system geometry, the K\"ahler potential, which encapsulates the geometric structure of the system, is given by the squared Hardy norm of the logarithmic transfer function. Since the logarithmic transfer function admits a series expansion in terms of complex cepstrum coefficients, its Hardy norm coincides with the complex cepstrum norm of the system. Similarly, the Hardy norm of the logarithmic power spectrum corresponds to the power cepstrum norm of the linear system.
	
	Since the K\"ahlerian information geometry discussed above has primarily been developed in the context of the $\alpha$-divergence-induced geometry of signal processing models, the relationship between $\alpha$-divergence and K\"ahler structures is by now relatively well understood \cite{choi2015kahlerian}. Importantly, the metric tensor of the K\"ahlerian information geometry is derived not only from the $\alpha$-divergence but also from the K\"ahler potential, which is given by the square of the unweighted cepstrum norm of the logarithmic transfer function \cite{choi2015kahlerian}.
	
	In contrast, the development of information geometry based on weighted Hardy norms has received little attention. Moreover, existing works on information geometry have not addressed the K\"ahlerization of other information manifolds. For instance, while the mutual information between past and future in ARMA processes where Martin \cite{martin2000metric} computed the associated metric tensor can be interpreted as a weighted complex cepstrum norm, the K\"ahler extension of the corresponding mutual information geometry remains unexplored.
	
	This paper extends the framework of K\"ahlerian information geometry to signal processing filters defined in weighted Hardy spaces. We establish that the geometry of linear systems, where composite functions formed by smooth transformations of the transfer function belong to weighted Hardy spaces, is a K\"ahler manifold equipped with an explicit Hermitian structure. On the resulting information manifold, the K\"ahler potential is given by the squared weighted Hardy norm of the transformed transfer function. In particular, when the logarithmic function is employed as the smooth transformation, the K\"ahler potential corresponds to the squared weighted complex cepstrum norm of the linear system.
	
	These results demonstrate a two-fold generalization of the K\"ahlerian information geometry \cite{choi2015kahlerian}: First, the generalization from the logarithmic transfer function to arbitrary smooth transformations $\phi$ of the transfer function, and second, the extension from the unweighted Hardy space to weighted Hardy spaces with arbitrary weight vectors $\omega$. The ($\phi$, $\omega$)-generalization of K\"ahlerian information geometry provides a unified framework that incorporates various information manifolds. Specifically, for linear systems in weighted Hardy spaces, a family of weight vectors not only induces the K\"ahlerian information geometry \cite{choi2015kahlerian} but also includes the geometry derived from the mutual information between past and future \cite{martin2000metric}.
	
	The structure of this paper is as follows. In the next section, we provide the basic definitions of weighted Hardy spaces, cepstrum, and an introduction to K\"ahler manifolds, which serve as the theoretical foundation for subsequent discussions. In Section \ref{sec_whardy_kahler}, we present the main theorems and corollaries regarding the information manifolds of linear systems in weighted Hardy spaces. Section \ref{sec_whardy_example} revisits several time series models within the framework of weighted Hardy spaces, examining the geometric properties of these linear systems. Finally, we conclude the paper in Section \ref{sec_whardy_conclusion}.

\section{Theoretical backgrounds}
\label{sec_whardy_theory}
	In this section, we first review the definitions of various function spaces in functional analysis and then generalize these spaces to weighted Hardy spaces. Following this, we revisit the concept of the complex cepstrum in signal processing and explore its connection to weighted Hardy norms. Additionally, we provide a brief overview of the properties of K\"ahler manifolds, which will serve as foundational concepts for the subsequent discussions. 
		
\subsection{Weighted Hardy spaces}
	A complex function $f(z)$ is given by the following Fourier series expansion for $-\pi\le w <\pi$:
	\begin{align}
		f(e^{jw})=\sum_{s=-\infty}^{\infty} f_s e^{-jsw},
	\end{align}
	where $f_s$ is the $s$-th Fourier coefficient of the series expansion. By using Z-transformation $(e^{jw}\to z)$, the function $f(z)$ in a complex domain is represented with
	\begin{align}
		f(z)=\sum_{s=-\infty}^{\infty} f_s z^{-s},
	\end{align}
	where $f_s$ is the $s$-th Fourier coefficient of $f$. In the opposite direction, we can transform a polynomial in $z$ to a discrete Fourier series by using inverse Z-transformation $(z\to e^{jw})$.
	
	Various function spaces are defined for functions on complex domains $\Omega$. The first example is the Lebegue space ($L^p$-space) for $1\le p<\infty$ that is the Banach space of functions with finite $L^p$-norms:
	\begin{align}
		L^p=\{f: \|f\|_{L^p}<\infty\},
	\end{align}  
	where the $L^p$-norm of a complex-valued function $f$ on $\Omega$ is defined as
	\begin{align}
		\|f\|_{L^p}=\bigg(\int_\Omega |f(z)|^p d\mu\bigg)^{1/p}.
	\end{align}
	For infinite $p$, $\|f\|_{L^\infty} = \esssup f$ where $\esssup$ is the essential supremum.
	
	In particular, the $p=2$ case is more intriguing than other $p$ values. The $L^2$-space is the function space with finite $L^2$-norms:
	\begin{align}
		L^2=\{f: \|f\|_{L^2}<\infty\}.
	\end{align} 
	The $L^2$-norm of a Fourier-transformed (or Z-transformed) function $f$ is given by
	\begin{align}
		\|f\|_{L^2}=\bigg(\sum_{s=-\infty}^{\infty} |f_s|^2\bigg)^{1/2},
	\end{align}
	where $f_s$ is the $s$-th Fourier (or Z-transformed) coefficient of $f$. The $L^2$-norm is also obtained from the inner product of two functions $f$ and $g$:
	\begin{align}
		\langle f,g \rangle=\frac{1}{2\pi}\int_{\Omega} f(z)\overline{g(z)} d\mu,
	\end{align}
	such that $\|f\|^2_{L^2}=\langle f,f \rangle$.
	
	Let us focus on analytic functions defined on the unit disk $\mathbb{D}$. In this setting, the Hardy space $H^p$ consists of analytic functions with finite $H^p$-norm, defined by
	\begin{align}
		\|f\|_{H^p}=\sup_{0<r<1}\bigg(\frac{1}{2\pi}\int_{-\pi}^{\pi}|f(r e^{i\theta})|^p d\theta\bigg)^{1/p},
	\end{align}
	where $1\le p<\infty$. For infinite $p$, $\|f\|_{H^\infty} = \sup_{|z|< 1} \left|f(z)\right|$.  
	
	For $f$ in $H^p$-spaces, it is possible to define the radial limit of the functions:
	\begin{align}
		\tilde{f}=\lim_{r\to1} f(r e^{i\theta})
	\end{align}
	for almost every $\theta$. Additionally, the relation between $H^p$-norms and $L^p$-norms is given by
	\begin{align}
		\| \tilde{f}\|_{L^p}=\| f \|_{H^p}.
	\end{align}
	
	 Similar to $L^p$-spaces, our interest in this paper is the $p=2$ case among various $H^p$-norms. For Fourier-transformed (or Z-transformed) functions $f$, the $H^2$-space is the function space of
	\begin{align}
		H^2=\{f: \|f\|_{H^2}<\infty\},
	\end{align}
	such that the $H^2$-norm of a Fourier-transformed (or Z-transformed) function is defined as
	\begin{align}
		\|f\|_{H^2}=\bigg(\sum_{s=0}^{\infty} |f_s|^2\bigg)^{1/2},
	\end{align}
	where $f_s$ is the $s$-th Fourier (or Z-transformed) coefficient of $f$.
	
	Similar to the $L^2$-space, the $H^2$-space is also equipped with the inner product. The inner product between two functions $f$ and $g$ in the Hardy spaces is given by
	\begin{align}
		\label{hardy_inner_product}
		\langle f,g \rangle=\frac{1}{2\pi i}\oint_{\mathbb{T}} f(z)\overline{g(z)} d\mu(z),
	\end{align}
	where $\mathbb{T}$ is the unit circle.
	
	Weighted Hardy spaces are defined as the generalization of the Hardy space such that weight vectors are introduced into the definition of the Hardy norm \cite{paulsen2016introduction}. Given a positive sequence  $\omega=(\omega_0,\omega_1,\cdots,\omega_s,\cdots)$ such that $\omega_s>0$ for all non-negative integers $s$, the weighted Hardy norm ($H^2_\omega$-norm) is represented with
	\begin{align}
		\|f\|_{H^2_\omega}=\|f\|_{\omega}=\bigg(\sum_{s=0}^{\infty} \omega_s |f_s|^2\bigg)^{1/2},
	\end{align}
	where $f_s$ is the $s$-th Fourier (or Z-transformed) coefficient of $f$. 
	
	The weighted Hardy space $H^2_\omega$ is defined as a function space of a finite weighted Hardy norm with a given weight vector $\omega$:
	\begin{align}
		H^2_\omega= \{f: \|f\|_\omega<\infty\}.
	\end{align}
	The reproducing kernel in weighted Hardy spaces is expressed with a given weight sequence $\omega$ \cite{paulsen2016introduction}:
	\begin{align}
		\label{rep_kernel_whardy}
		k_u(v)=\sum_{s=0}^{\infty} \frac{\bar{u}^s v^s}{\omega_s},
	\end{align}
	where $u,v$ are in the unit disk $\mathbb{D}$. 
	
	It is obvious that the unweighted Hardy space is a special case of weighted Hardy spaces with the weight vector of the unit sequence, i.e., $\omega=(1,1,\cdots)$. The reproducing kernel in the unweighted Hardy space with the weight vector of the unit sequence $\omega=(1,1,\cdots)$ is given by
	\begin{align}
		k_u(v)=\sum_{s=0}^{\infty} \bar{u}^sv^s=\frac{1}{1- \bar{u}v},
		\label{rep_kernel_uw}
	\end{align}
	and this reproducing kernel is also known as the Sz\H{e}go kernel.
		
	Besides the unweighted Hardy space, we can consider the weighted Hardy spaces with $\omega_s=|\rho|^{2s}$. This weight corresponds to the exponentiation in Z-transformation meaning the scaling of $z^{-1} \to \rho z^{-1}$ in $z$-domain. This exponentiation also changes the region of convergence by a factor of $\rho$. The unweighted Hardy norm of a signal filter with the exponentiation is equivalent to the weighted Hardy norm of the signal filter. The reproducing kernel of the weighted Hardy space with the weight vector of the unit sequence $\omega_s=|\rho|^{2s}$ is represented by
	\begin{align}
		k_u(v)=\sum_{s=0}^{\infty} \frac{\bar{u}^sv^s}{|\rho|^{2s}}=\frac{1}{1- \bar{u}v/|\rho|^2},
		\label{rep_kernel_expw}
	\end{align}
	and it is straightforward to check that Eq. (\ref{rep_kernel_expw}) converges to Eq. (\ref{rep_kernel_uw}) in the limit of $|\rho| \to 1$.
	
	Other function spaces are also interpreted as weighted Hardy spaces. Another examples is the Sobolev space. The Sobolev space $\mathcal{W}^{m,p}$ for an integer $m$ and $1\le p\le\infty$ is defined as
	\begin{align}
		\mathcal{W}^{m,p}=\{f: f \in L^p, f^{(l)}\in L^p \textrm{ for }  l\le m \},
	\end{align}
	where $f^{(l)}$ is the $l$-th derivative of $f$. This space is the function space of finite Sobolev norms. 
	
	Similar to the Hardy space, analytic functions on the unit disk with $p=2$ are our main interest. In this case, the Sobolev norm of a Fourier-transformed function $f$ is expressed with
	\begin{align}
		\|f\|_{\mathcal{W}^{m,2}}=\bigg(\sum_{s=0}^{\infty} (1+s^2+s^4+\cdots+s^{2m}) |f_s|^2\bigg)^{1/2},
	\end{align}
	where $f_s$ is the $s$-th Fourier coefficient. It is obvious that the Sobolev space $\mathcal{W}^{m,2}$ is the weighted Hardy space of $\omega_s=1+s^2+s^4+\cdots+s^{2m}$.
	
	The Dirichlet space is also an example of weighted Hardy spaces. The Dirichlet space $\mathcal{D}$ is defined as the function space of the Dirichlet semi-norm given by
	\begin{align}
		\|f\|_{\mathcal{D},*}=\bigg(\frac{1}{\pi}\int_{\Omega}|f'(z)|^2d\mu\bigg)^{1/2},
	\end{align}
	where $\Omega$ is the unit disk $\mathbb{D}$ in this paper. By the definition of the Dirichlet semi-norm given above, the semi-norm of a Fourier-transformed function $f$ is represented with the Fourier coefficients in the form of
	\begin{align}
		\|f\|_{\mathcal{D},*}=\bigg(\sum_{s=1}^{\infty} s|f_s|^2\bigg)^{1/2},
	\end{align}
	and it is straightforward to confirm that the weight vector of the Dirichlet space is given by $\omega_s=s$ for a positive integer $s$. 
	
	By plugging the weight vector of $\omega_s=s$ into Eq. (\ref{rep_kernel_whardy}), the reproducing kernel in the Dirichlet semi-norm space is given by 
	\begin{align}
		k_u(v)=\log{(1-\bar{u}v)},
	\end{align}
	where $u$ and $v$ are complex numbers in the unit disk $\mathbb{D}$. 
	
	Since the semi-norm is independent of the 0-th Fourier coefficient, all constant functions belong to the identical zero norm. Several ways to extend the semi-norm to the Dirichlet norm are as follows. The simplest solution to the identical Dirichlet semi-norm issue is including the 0-th Fourier coefficient $f_0$ into the semi-norm:
	\begin{align}
		\|f\|_{\mathcal{D}}=\bigg(|f_0|^2+\|f\|^2_{\mathcal{D},*}\bigg)^{1/2}.
	\end{align}
	Another solution is adding the Hardy norm to the Dirichlet semi-norm. By adding the Hardy norm, the Dirichlet norm is given by
	\begin{align}
		\|f\|_{\mathcal{D}}=\bigg(\|f\|^2_{H^2}+\|f\|^2_{\mathcal{D},*}\bigg)^{1/2}=\bigg(\sum_{s=0}^{\infty} (1+s)|f_s|^2\bigg)^{1/2},
	\end{align}
	and the reproducing kernel in the Dirichlet space is found as
	\begin{align}
		k_u(v)=\frac{1}{\bar{u}v}\log{(1-\bar{u}v)},
	\end{align}
	where $u, v \in\mathbb{D}\setminus \{0\}$. 
	
	The Bergman space is also a weighted Hardy space. The Bergman space $\mathcal{A}$ is the function space of a finite Bergman norm defined as
	\begin{align}
		\|f\|_{\mathcal{A}}=\bigg(\frac{1}{\pi}\int_{\Omega}|f(z)|^2d\mu\bigg)^{1/2},
	\end{align}
	where $\Omega$ is given by the unit disk $\mathbb{D}$ in this paper. It is noteworthy that the Bergman norm uses $f(z)$ instead of $f'(z)$ in the Dirichlet semi-norm. The Bergman norm is represented with the Fourier coefficients $f_s$:
	\begin{align}
		\|f\|_{\mathcal{A}}=\bigg(\sum_{s=0}^{\infty}\frac{|f_s|^2}{1+s}\bigg)^{1/2},
	\end{align}
	such that the weight sequence is $\omega_s = (1+s)^{-1}$ for non-negative $s$. From Eq. (\ref{rep_kernel_whardy}) and the weight sequence, the Bergman kernel is straightforwardly obtained as 
	\begin{align}
		k_{u}(v)=\frac{1}{(1-\bar{u}v)^2},
	\end{align}
	where $u$ and $v$ are complex numbers inside the unit disk $\mathbb{D}$. 
	
	An interesting property of the Bergman kernel is that the Bergman metric is emergent from the Bergman kernel:
	\begin{align}
		g_{v\bar{u}}=\partial_v\partial_{\bar{u}}\log{k_{u}(v)}=\frac{2}{(1-\bar{u}v)^2},
	\end{align}	
	where un-barred and barred indices are holomorphic and anti-holomorphic coordinates, respectively.
	
	Similar to the Dirichlet space and the Sobolev space, function spaces of differentiation-related semi-norms can be introduced. The generalized differentiation semi-norm is given by
	\begin{align}
		\|f\|_{\tilde{\mathcal{D}}^m,*}=\bigg(\frac{1}{2\pi}\int_{-\pi}^\pi \bigg| \bigg(\frac{d}{d\theta}\bigg)^{\frac{m}{2}} f(e^{i\theta})\bigg|^2 d\theta\bigg)^{1/2},
	\end{align}
	where the fractional derivative is defined as
	\begin{align}
		\bigg(\frac{d}{d\theta}\bigg)^{m} e^{is\theta}=(is)^m e^{is\theta}
	\end{align}
	for real numbers $m$ and $s$. With Z-transformation, the generalized differentiation semi-norm is represented with
	\begin{align}
		\|f\|_{\tilde{\mathcal{D}}^m,*}=\bigg(\frac{1}{2\pi}\int_{-\pi}^\pi \bigg| \bigg(z\frac{d}{dz}\bigg)^{\frac{m}{2}} f(z)\bigg|^2 \frac{dz}{z}\bigg)^{1/2},
	\end{align}
	where the fractional derivative is defined as
	\begin{align}
		\bigg(z\frac{d}{dz}\bigg)^{m} z^s= s^m z^s
	\end{align}
	for real numbers $m$ and $s$. Although a negative real number $m$ actually corresponds to integration, we keep the convention for the integration as the generalized differentiation with a negative order in order to pursue consistency in notation. 
	
	The generalized differentiation semi-norm is expressed with the Fourier (or Z-transformed) coefficients:
	\begin{align}
		\|f\|_{\tilde{\mathcal{D}}^m,*}=\bigg(\sum_{s=0}^{\infty} s^m |f_s|^2\bigg)^{1/2},
	\end{align}
	and this semi-norm is regarded the weighted Hardy norm with the weight vector of $\omega_s=s^m$. It is straightforward to obtain the reproducing kernel in the generalized differentiation semi-norm function space as
	\begin{align}
		k_u^{(m)}(v)=Li_m(\bar{u}v),
	\end{align}
	where $u,v$ are complex numbers inside the unit disk $\mathbb{D}$, and $Li_m$ is the polylogarithm of order $m$ such that
	\begin{align}
		Li_m(z)=\sum_{s=1}^{\infty}\frac{z^s}{s^m}
	\end{align}
	for $|z|<1$. The polylogarithm is extended to $|z|=1$ by analytic continuation.
	
	The generalized differentiation semi-norm is cross-connected to various function spaces. First of all, the $m=0$ case corresponds to the unweighted Hardy space. When $m=1$, the generalized differentiation semi-norm is identical to the Dirichlet semi-norm. In addition, the Sobolev norm is also obtained from the sum of the generalized differentiation semi-norms:
	\begin{align}
		\|f\|_{\mathcal{W}^{m,2}}=\bigg(\sum_{l=0}^{m} \|f\|^2_{\tilde{\mathcal{D}}^{2l},*}\bigg)^{1/2}.
	\end{align}
	Similarly, various weighted Hardy norms are constructed from positive linear combinations of the generalized differentiation semi-norms with different $m$ values.
	
	As a summary, the weighted Hardy norms introduced in this subsection are represented as
	\begin{align}
		\|f\|_{\omega}=\bigg(\sum_{s=0}^{\infty} \omega_s |f_s|^2\bigg)^{1/2},
	\end{align}
	where $\omega_s$ is the weight sequence of a given function space. For example, the weight sequences $\omega$ of the well-known function spaces covered in this section are followings:
	\begin{align}
		\omega_s=\left\{ 
		\begin{array}{ll}
			1 & \textrm{for unweighted Hardy space } H^2\\ 
			1+s^2+s^4+\cdots+s^{2m} & \textrm{for Sobolev space } \mathcal{W}^{m,2}\\
			s & \textrm{for Dirichlet space } \mathcal{D}\\
			\frac{1}{1+s} & \textrm{for Bergman space } \mathcal{A}\\
			s^{m} & \textrm{for differentiation semi-norm space } \tilde{\mathcal{D}}^m
		\end{array}
	\right.
	\end{align}
	for a non-negative integer $s$.
	
\subsection{Cepstrum and weighted Hardy norms}
	In signal processing, the transfer function of a filter encodes essential information about the system, as it characterizes how input signals are transformed into output signals. In the frequency domain, the Fourier-transformed transfer function in the frequency domain $w$ is given by
	\begin{align}
		h(w;\boldsymbol{\xi})=\sum_{s=-\infty}^{\infty} h_s(\boldsymbol{\xi}) e^{-jsw},
	\end{align}
	where $h_s$ is the $s$-th Fourier coefficient of the filter transfer function, and $\boldsymbol{\xi}$ is the vector of signal filter parameters. 
	
	Using the convention of $z^{-s}$ in Z-transformation, the transfer function of a causal filter takes the unilateral form involving only terms with non-negative $s$. Specifically, the transfer function of a causal filter is expressed in the following unilateral form:
	\begin{align}
		h(z;\boldsymbol{\xi})=\sum_{s=0}^{\infty} h_s(\boldsymbol{\xi}) z^{-s},
	\end{align}
	and the transfer function is analytic outside the unit disk $\mathbb{D}$. The various function norms defined in the previous subsection can be applied to transfer functions in $z^{-s}$ after the inverse projection of $z\to z^{-1}$. In another way, the transfer functions are expressed in $z^{s}$ convention and the function norms are computed.
	
	Hardy spaces and the associated Hardy norms can play a fundamental role in characterizing linear systems in signal processing. For instance, the stationarity condition of a linear system can be expressed with the unweighted Hardy norm of its transfer function:
	\begin{align}
		\sum_{s=0}^{\infty}|h_s|^2=\|h(z;\boldsymbol{\xi})\|^2_{H^2}<\infty,
	\end{align}
	and this condition indicates that the transfer functions of stationary filters are functions in the unweighted Hardy space.
	
	Additionally, the concepts of the Hardy spaces and the Hardy norms are applicable to the parameter estimation of linear systems. The $H^p$-norms can be used to define the distance between two transfer functions, facilitating the identification of optimal system parameters $\boldsymbol{\xi}^*$ under given constraints $\mathcal{C}$:
	\begin{align}
		\boldsymbol{\xi}^{*}=\overset{}{\underset{\mathrm{\boldsymbol{\xi}\in\mathcal{C}}}{\textrm{argmin}}} \| h(z;\boldsymbol{\xi})-\hat{h}\|_{H^p},
	\end{align}
	where $\hat{h}$ is the target transfer function and $\mathcal{C}$ is the parameter set satisfying the constraints. A notable application is $H^{\infty}$-optimization in control theory \cite{chui1997discrete}.
	
	As generalization, we introduce a composite function $f$ of a smooth transformation $\phi$ and a transfer function $h$ such that $f=\phi\circ h:\mathbb{D}\to \mathbb{C}$ where $f$ is analytic in the unit disk $\mathbb{D}$. The weighted Hardy norm of a linear system described by $f$ is given by
	\begin{align}
	\label{weighted_hardy_general_norm}
		\mathcal{I}_{\omega}=\| \phi\circ h \|_{\omega}=\bigg(\sum_{s=0}^{\infty} \omega_s | f_s(\boldsymbol{\xi}) |^2\bigg)^{1/2},
	\end{align}
	where $f_s$ is the $s$-th Fourier coefficient of $f=\phi\circ h$. 
	
	Similarly, the distance between two linear systems $M_1$ and $M_2$, of which the transfer functions are $h_1(z;\boldsymbol{\xi}')$ and $h_2(z;\boldsymbol{\xi})$, respectively, is found as
	\begin{align}
		\label{weighted_hardy_general_dist}
		\mathcal{I}_{\omega}(M_1, M_2)=\| \phi\circ h_1-\phi\circ h_2 \|_{\omega}=\bigg(\sum_{s=0}^{\infty} \omega_s | f_{1,s}(\boldsymbol{\xi}')-f_{2,s}(\boldsymbol{\xi}) |^2\bigg)^{1/2},
	\end{align}
	where $f_{i,s}$ is the $s$-th Fourier coefficient of $f_i$.   
	
	It is noteworthy that the weighted Hardy norm-based distance of Eq. (\ref{weighted_hardy_general_dist}) is a distance measure in metric spaces. First of all, if two linear systems are identical, the weighted Hardy distance between the two systems is zero. Moreover, it is symmetric under exchange of $M_1$ and $M_2$ in Eq. (\ref{weighted_hardy_general_dist}). Furthermore, the triangle inequality is also satisfied by this distance measure.
	
	Among various candidate functions to $\phi$, the most intriguing transformation, which is also applicable to signal processing, is the logarithmic function. In this case, the logarithmic transfer function corresponds to the complex cepstrum of a linear system \cite{oppenheim1965superposition}. By using the complex cepstrum and Fourier/Z-transformations, the logarithmic transfer function of a linear system can be expressed with a power series of $z$:
	\begin{align}
		\label{log_transfer_cepstrum}
		\log{h(z;\boldsymbol{\xi})}=\sum_{s=0}^{\infty} c_s(\boldsymbol{\xi})  z^{-s},
	\end{align}
	where $c_s$ is the $s$-th complex cepstrum coefficient of the transfer function. The $s$-th complex cepstrum coefficient $c_s$ in Eq. (\ref{log_transfer_cepstrum}) is found directly from  the inner product in Hardy spaces of Eq. (\ref{hardy_inner_product}) with the basis of $\{1, z^{-1}, z^{-2}, \cdots, z^{-s}, \cdots\}$:
	\begin{align}
		\label{cepstrum_coeff}
		c_s(\boldsymbol{\xi})=\langle \log{h(z,\boldsymbol{\xi})}, z^{-s} \rangle=\frac{1}{2\pi i}\oint_{\mathbb{T}} \log{h(z,\boldsymbol{\xi})} z^{s}\frac{dz}{z},
	\end{align}
	where $\mathbb{T}$ is the unit circle.
	
	Various function norms of logarithmic transfer functions can be interpreted within the framework of weighted Hardy spaces. For instance, the complex cepstrum norm commonly used in signal processing corresponds to the unweighted Hardy norm of a logarithmic transfer function. As a natural extension of Eq. (\ref{weighted_hardy_general_norm}) and Eq. (\ref{weighted_hardy_general_dist}), weighted complex cepstrum norms are identified with the weighted Hardy norms of logarithmic transfer functions. Leveraging the homomorphic structure of the cepstrum \cite{proakisdigital}, the weighted complex cepstrum distance between two linear systems $M_1$ and $M_2$, with transfer functions $h_1$ and $h_2$, respectively, is given by
	\begin{align}
		\label{cepstrum_norm}
		\mathcal{I}_{\omega}(M_1,M_2)=\|\log{h_1}-\log{h_2}\|_{\omega}=\bigg(\sum_{s=0}^{\infty} \omega_s |c_{1,s}(\boldsymbol{\xi}')-c_{2,s}(\boldsymbol{\xi})|^2\bigg)^{1/2},
	\end{align}
	where $c_{i,s}$ is the $s$-th complex cepstrum coefficient of a linear system $M_i$. Additionally, the weighted complex cepstrum distance is also a distance measure in metric spaces as discussed above. In particular, when $h_2=1$, the weighted complex cepstrum norm of a transfer function can be interpreted as the weighted complex cepstrum distance between the input signal and the output signal produced by the corresponding filter. In this context, the weighted Hardy norm of the logarithmic transfer function serves as a quantitative measure of how the output differs from the input, effectively characterizing the system through this distance metric.

\subsection{K\"ahler geometry}
	In this subsection, we briefly review the fundamental concepts of K\"ahler manifolds that will be used in subsequent derivations. A more detailed introduction to K\"ahler manifolds can be found in reference \cite{nakahara2003geometry,choi2015application}.
	
	By definition, a K\"ahler manifold is a Hermitian manifold whose associated K\"ahler form is closed. More precisely, a complex manifold with the real dimension of $2n$ is K\"ahler if and only if its metric tensor components satisfy the following conditions:
\begin{align} 
\label{metric_con_hermitian}
	g_{ij}&=g_{\bar{\imath}\bar{j}}=0,\\
	\partial _{i}g_{j\bar{k}}=\partial _{j}g_{i\bar{k}},&\quad\partial _{\bar{\imath}}g_{k\bar{j}}=\partial _{\bar{j}}g_{k\bar{\imath}},
\label{metric_con_closed_kahler}
\end{align}
where barred and unbarred indices are the holomorphic and anti-holomorphic coordinates of the K\"ahler manifold, and $i,j,k=1,\cdots,n$. Eq. (\ref{metric_con_hermitian}) is the condition for a Hermitian manifold, and Eq. (\ref{metric_con_closed_kahler}) is for a closed K\"ahler form.

	K\"ahler manifolds take several computational advantages over non-K\"ahler manifolds. First, non-trivial components of metric tensor and Levi-Civita connection are calculated from the K\"ahler potential:
\begin{align}
\label{metric_kahler_rel}
g_{i\bar{j}}&=(g_{\bar{i}j})^*=\partial _{i}\partial _{\bar{j}}\mathcal{K},\\
\label{connection_kahler_rel}
\Gamma _{ij,\bar{k}}&=(\Gamma _{\bar{i}\bar{j},k})^*=\partial _{i}\partial _{j}\partial _{\bar{k}}\mathcal{K},
\end{align}
where $\mathcal{K}$ is the K\"ahler potential. Eq. (\ref{metric_kahler_rel}) and Eq. (\ref{connection_kahler_rel}) indicate significant benefits in geometric calculation on K\"ahler manifolds. If the K\"ahler potential is given, the components of the metric tensor and Levi-Civita connection are easily calculated from taking partial derivatives with respect to only relevant coordinates. In non-K\"ahler geometry, these two geometric objects require more lengthy and tedious computation steps with partial derivatives and summations across all the coordinates of the manifold. 

	Moreover, the Ricci tensor of K\"ahler manifolds is also straightforwardly obtained in a simpler way similar to Eq. (\ref{metric_kahler_rel}) and Eq. (\ref{connection_kahler_rel}). The Ricci tenor of K\"ahler manifolds is given by
\begin{align}  
	\label{ricci_kahler_rel}
	R_{i\bar{j}}=-\partial_i\partial_{\bar{j}}\log{\mathcal{G}},
\end{align}
where $\mathcal{G}$ is the determinant of the metric tensor Eq. (\ref{metric_kahler_rel}). The Ricci tensor of non-K\"ahler manifolds needs much more complicated calculation procedures such as additional calculation steps for Riemann curvature tensor than Eq. (\ref{ricci_kahler_rel}). 

	Furthermore, the Laplace-Beltrami operator in K\"ahler geometry is in a more straightforward form of
\begin{align}
	\Delta =2g^{i\bar{j}}\partial _{i}\partial _{\bar{j}},
\end{align}
where $g^{i\bar{j}}$ is the component of the inverse metric tensor. It is obvious that this Laplace-Beltrami operation in K\"ahler geometry is much easier to calculate than that of non-K\"ahler manifolds in the following form:
\begin{align}
	\Delta = \frac{1}{\sqrt{\mathcal{G}}}\partial_i \big( \sqrt{\mathcal{G}} g^{i\bar{j}}\partial_{\bar{j}}\big),
\end{align}
where $\mathcal{G}$ is the determinant of the metric tensor and $g^{i\bar{j}}$ is the component of the inverse metric tensor. 

\section{Geometry of linear systems in weighted Hardy spaces}
\label{sec_whardy_kahler}
	In this section, we derive the geometric structure of linear systems in weighted Hardy spaces. In particular, we prove that the manifolds induced by weighted Hardy norms, applied to various smooth transformations of transfer functions, are K\"ahler manifolds. Furthermore, we show that the K\"ahler potential of the resulting signal filter geometry is given by the square of the weighted Hardy norm of the composite function formed by the smooth transformation and the transfer function.
	
	Let us denote a composite function $f$ of a smooth transformation $\phi$ and the transfer function of a linear system $h(z;\boldsymbol{\xi})$, i.e., $f=\phi \circ h$. We can consider the weighted Hardy norm of $f$: 
	\begin{align}
		\mathcal{I}_{\omega}=\bigg(\sum_{s=0}^{\infty} \omega_s |f_s(\boldsymbol{\xi})|^2\bigg)^{1/2}.
	\end{align}
	Induced from this weighted Hardy norm, we are able to derive the geometry of signal processing filters in weighted Hardy spaces by the following theorem.
	\begin{thm}
	\label{thm_kahler_whardy}
		Let $h(z;\boldsymbol{\xi})$ be the transfer function of a linear system, and let $\phi$ be a smooth transformation such that the composite function $f=\phi \circ h$ belongs to a weighted Hardy space. Then, the information geometry induced by the weighted Hardy norm of $\phi \circ h$ defines a K\"ahler manifold with the corresponding parameter space of the linear system.
	\end{thm}
	\begin{proof}
	First of all, we can consider infinitesimal distance for deriving metric tensor. Since metric tensor components on a Riemannian manifold are extracted from an infinitesimal length, the infinitesimal weighted Hardy norm of the linear system geometry with a given weight sequence $\omega$ is straightforwardly obtained from Eq. (\ref{weighted_hardy_general_dist}):
	\begin{align}
		\label{whardy_length_element}
		\delta \mathcal{I}^2_{\omega}=\sum_{s=0}^{\infty} \omega_s |\delta f_s(\boldsymbol{\xi})|^2,
	\end{align}
	where $\delta f_s$ is expanded in terms of the infinitesimal displacements $\delta \xi^i$ along the $i$-th coordinate $\xi^i$ of the $n$-dimensional manifold. Metric tensor components of a Riemannian manifold are decoded from the quadratic terms of the infinitesimal displacements $\delta \xi^i$ in the length element Eq. (\ref{whardy_length_element}). 
	
	Alternatively, the metric tensor and connection components of the geometry induced from a divergence $\mathcal{D}$ \cite{amari2000methods} are represented as
	\begin{align}
		\label{ig_metric}
		g_{\mu\nu}&=-\mathcal{D}(\partial_\mu, \tilde{\partial}_\nu)|_{\boldsymbol{\xi}=\tilde{\boldsymbol{\xi}}},\\
		\label{ig_connection}
		\Gamma_{\mu\nu,\rho}&=-\mathcal{D}(\partial_\mu \partial_\nu,\tilde{\partial}_\rho)|_{\boldsymbol{\xi}=\tilde{\boldsymbol{\xi}}},
	\end{align}
	where $\mu, \nu, \rho$ run for the coordinates of the geometry.
	
	Using the squared weighted Hardy norm $\mathcal{I}_{\omega}^2$ as $\mathcal{D}$ in Eq. (\ref{ig_metric}), the metric tensor components of the weighted Hardy norm-induced geometry of a linear system are given by
	\begin{align}
		\label{metric_whardy1}
		g_{ij}(\boldsymbol{\xi},\bar{\boldsymbol{\xi}};\omega)&=\big(g_{\bar{i}\bar{j}}(\boldsymbol{\xi},\bar{\boldsymbol{\xi}};\omega)\big)^{*}=0,\\
		\label{metric_whardy2}
		g_{i\bar{j}}(\boldsymbol{\xi},\bar{\boldsymbol{\xi}};\omega)&=\big(g_{\bar{i}j}(\boldsymbol{\xi},\bar{\boldsymbol{\xi}};\omega)\big)^{*}=\sum_{s=0}^{\infty} \omega_s \partial_i f_s (\boldsymbol{\xi}) \partial_{\bar{j}} \bar{f}_s (\bar{\boldsymbol{\xi}}),
	\end{align}
	where $i,j=1,\cdots,n$. By Eq. (\ref{metric_whardy1}), it is straightforward to check that the geometry from a given weighted Hardy norm is the Hermitian manifold.
	
	Similarly, the connection components of the geometry are calculated as
	\begin{align}
		\label{connection_whardy}
		\Gamma_{ij,\bar{k}}(\boldsymbol{\xi},\bar{\boldsymbol{\xi}};\omega)&=(\Gamma_{\bar{i}\bar{j},k}(\boldsymbol{\xi},\bar{\boldsymbol{\xi}};\omega))^*=\sum_{s=0}^{\infty} \omega_s \partial_i\partial_{j} f_s (\boldsymbol{\xi}) \partial_{\bar{k}} \bar{f}_s (\bar{\boldsymbol{\xi}}),
	\end{align}
	where $i,j,k=1,\cdots,n$. All other components of the connection are vanishing.
	
	The Levi-Civita connection of the manifold is represented with 
	\begin{align}
		\Gamma_{\mu\nu,\rho} = \frac{1}{2}(\partial_{\mu}g_{\nu\rho}+\partial_{\nu}g_{\mu\rho}-\partial_{\rho}g_{\mu\nu}),
	\end{align}
	where $\mu,\nu,\rho$ run through holomorphic and anti-holomorphic coordinates, i.e., $\mu,\nu,\rho=1,\cdots,n,\bar{1},\cdots,\bar{n}$. From the metric tensor expressions of Eq. (\ref{metric_whardy1}) and Eq. (\ref{metric_whardy2}), the non-trivial components of the Levi-Civita connection are found as
	\begin{align}
		\label{connection_lc}
		\Gamma_{ij,\bar{k}}(\boldsymbol{\xi},\bar{\boldsymbol{\xi}};\omega)&=(\Gamma_{\bar{i}\bar{j},k}(\boldsymbol{\xi},\bar{\boldsymbol{\xi}};\omega))^*=\sum_{s=0}^{\infty} \omega_s \partial_i\partial_{j} f_s (\boldsymbol{\xi}) \partial_{\bar{k}} \bar{f}_s (\bar{\boldsymbol{\xi}}),
	\end{align}
	where $i,j,k=1,\cdots,n$. All other components of the Levi-Civita connection are vanishing.
	
	Eq. (\ref{metric_whardy1}) shows that the geometry induced from a weighted Hardy norms is a Hermitian manifold of Eq. (\ref{metric_con_hermitian}). It is also straightforward to verify that metric tensor components of Eq. (\ref{metric_whardy2}) satisfy the closed K\"ahler two-form condition of Eq. (\ref{metric_con_closed_kahler}).
		
	Since the metric tensor components, Eq. (\ref{metric_whardy1}) and Eq. (\ref{metric_whardy2}), of the geometry fulfill Eq. (\ref{metric_con_hermitian}) and Eq. (\ref{metric_con_closed_kahler}), the information manifold of a linear system is the Hermitian manifold with a closed K\"ahler two-form. By the definition of K\"ahler manifolds, the geometry of linear systems in weighted Hardy spaces is a K\"ahler manifold.
	\end{proof}
	
	In Theorem \ref{thm_kahler_whardy}, the correspondence between K\"ahler manifolds and linear system manifolds is found by using the metric tensor components of Eq. (\ref{metric_whardy1}) and Eq. (\ref{metric_whardy2}), and the connection components of Eq. (\ref{connection_whardy}) on a given linear system manifold where smooth transformations of a transfer function are in weighted Hardy spaces. 
	
	Since metric tensor components and Levi-Civita connection components in K\"ahler geometry are derived directly from the K\"ahler potential, finding the K\"ahler potential of the linear system geometry is our next main question. The following theorem provides the answer.
	
	\begin{thm}
	\label{thm_kahler_potential_whardy}
		The K\"ahler potential of the information manifold associated with a linear system in a weighted Hardy space is given by the square of the weighted Hardy norm of a composite function formed by a smooth transformation and the system's transfer function.
	\end{thm}
	\begin{proof}
		By using the product rule of derivative to Eq. (\ref{metric_whardy2}), the metric tensor of K\"ahler geometry expressed with Fourier transformation (or Z-transformation) coefficients and a weight sequence is rewritten as
		\begin{align}
			g_{i\bar{j}}(\boldsymbol{\xi},\bar{\boldsymbol{\xi}};\omega)=\partial_i \partial_{\bar{j}}\bigg(\sum_{s=0}^{\infty} \omega_s |f_s (\boldsymbol{\xi}) |^2\bigg)=\partial_i \partial_{\bar{j}}\Big(\|f(z;\boldsymbol{\xi})\|^2_{\omega}\Big).
		\end{align}
		According to Eq. (\ref{metric_kahler_rel}) stating the relation between the metric tensor and the K\"ahler potential, the K\"ahler potential $\mathcal{K}$ is given by
		\begin{align}
			\label{kahler_whardy_norm}
			\mathcal{K}=\sum_{s=0}^{\infty} \omega_s |f_s (\boldsymbol{\xi}) |^2=\|f(z;\boldsymbol{\xi})\|^2_{\omega},
		\end{align}
		up to purely holomorphic and purely anti-holomorphic functions. Since  the metric tensor is independent of the purely holomorphic and purely anti-holomorphic terms, these terms are auxiliary in the definition of the K\"ahler potential. Eq. (\ref{kahler_whardy_norm}) indicates that the K\"ahler potential is the square of a weighted Hardy norm of a smooth transformation of the transfer function.
		
		We also obtain the identical K\"ahler potential from the connection.  By applying the product rule of derivative to Eq. (\ref{connection_whardy}), the connection components are expressed in the following form:
		\begin{align}
			\Gamma_{ij,\bar{k}}(\boldsymbol{\xi},\bar{\boldsymbol{\xi}};\omega)=\sum_{s=0}^{\infty} \omega_s  \partial_i\partial_{j} f_s \partial_{\bar{k}}\bar{f}_s=\partial_i\partial_{j} \partial_{\bar{k}}\bigg(\sum_{s=0}^{\infty} \omega_s  |f_s|^2\bigg)=\partial_i\partial_{j} \partial_{\bar{k}}\mathcal{K}.
		\end{align}
		It is obvious that the K\"ahler potential found above is identical to Eq. (\ref{kahler_whardy_norm}). Starting from $\Gamma_{\bar{i}\bar{j},k}$, it is also possible to obtain the same results for the K\"ahler potential. 
	\end{proof}
	
	By Theorem \ref{thm_kahler_potential_whardy}, it is concluded that the square of a weighted Hardy norm for a composite function of a smooth transformation and the transfer function is the K\"ahler potential from which the metric tensor and the Levi-Civita connection are derived by using partial derivatives. As mentioned in the previous section, this is an advantage of the K\"ahler manifolds that the geometric objects are obtained from the K\"ahler potential by simply taking holomorphic and anti-holomorphic derivatives. Using the metric tensor, the Ricci tensor in the induced geometry is also obtained from Eq. (\ref{ricci_kahler_rel}).
	
	Theorem \ref{thm_kahler_whardy} and Theorem \ref{thm_kahler_potential_whardy} can be applied to linear systems with the stationarity condition of a signal filter.
	
	\begin{crl}
	\label{crl_kahler_weighted_stat_con}
		A linear system geometry induced by a finite weighted stationarity condition is a K\"ahler manifold. The corresponding K\"ahler potential of the linear system geometry is the square of the weighted Hardy norm of the transfer function of a linear system, i.e., the weighted stationarity condition.
	\end{crl}
	
	\begin{proof}
		If $f(z;\boldsymbol{\xi})=h(z;\boldsymbol{\xi})$, i.e., $\phi(t)=t$, the weighted Hardy norm of $f$ is identical to the weighted Hardy norm of the transfer function, i.e., the weighted stationarity condition of the linear system. By Theorem \ref{thm_kahler_whardy}, the geometry induced from the weighted stationarity condition is a K\"ahler manifold.
		
		Plugging $f(z;\boldsymbol{\xi})=h(z;\boldsymbol{\xi})$ to Eq. (\ref{kahler_whardy_norm}), the K\"ahler potential of the geometry induced from the weighted Hardy norms of the transfer function of the linear system is given by
		\begin{align}
			\mathcal{K}=\|h(z;\boldsymbol{\xi})\|^2_{\omega}=\sum_{s=0}^{\infty} \omega_s |h_s (\boldsymbol{\xi}) |^2,
		\end{align}
		where $h_s$ is the $s$-th Fourier (Z-transformed) coefficient of the linear system.
	\end{proof}
		
	Considering that the complex cepstrum norm is a special case of weighted Hardy norms, Theorem \ref{thm_kahler_whardy} and Theorem \ref{thm_kahler_potential_whardy} are also applicable to linear systems with weighted complex cepstrum norms.
	
	\begin{crl}
	\label{crl_kahler_weighted_cepstrum}
		A linear system geometry induced by a finite weighted complex cepstrum norm is a K\"ahler manifold. The corresponding K\"ahler potential of the linear system geometry is the square of the weighted complex cepstrum norm of a linear system. In other words, the K\"ahler potential of the complex cepstrum geometry is the square of the weighted Hardy norm of the logarithmic transfer function.
	\end{crl}
	
	\begin{proof}
		Similar to Corollary \ref{crl_kahler_weighted_stat_con}, this is also a special case of Theorem \ref{thm_kahler_whardy} and Theorem \ref{thm_kahler_potential_whardy}. If $f(z;\boldsymbol{\xi})=\log{h(z;\boldsymbol{\xi})}$, i.e., $\phi(t)=\log{t}$, the weighted Hardy norm of $f$ is identical to the weighted complex cepstrum norm of the linear system. According to Theorem \ref{thm_kahler_whardy}, the geometry induced from the weighted complex cepstrum norm is a K\"ahler manifold.
		
		Plugging $f(z;\boldsymbol{\xi})=\log{h(z;\boldsymbol{\xi})}$ to Eq. (\ref{kahler_whardy_norm}), the K\"ahler potential of the geometry induced from the weighted complex cepstrum norm is given by
		\begin{align}
			\label{kahler_cepstrum}
			\mathcal{K}=\|\log{h(z;\boldsymbol{\xi}})\|^2_{\omega}=\sum_{s=0}^{\infty} \omega_s |c_s (\boldsymbol{\xi}) |^2,
		\end{align}
		where $c_s$ is the $s$-th complex cepstrum coefficient of the linear system.
	\end{proof}
	
	The theorems and corollaries proven above are the two-fold generalization of the corresponding the theorems and corollaries for the K\"ahlerian information geometry of a linear system in the unweighted Hardy space \cite{choi2015kahlerian}. First of all, not being limited to the logarithmic function in Corollary \ref{crl_kahler_weighted_cepstrum} and Choi and Mullhaupt \cite{choi2015kahlerian}, the introduction of more generic transformation functions $\phi$ also induces the K\"ahler structures. When the logarithmic function is used, the information manifold of a linear system in weighted Hardy space includes the K\"ahlerian information geometry given in Choi and Mullhaupt \cite{choi2015kahlerian}. Second, it is extended from the unweighted Hardy space to weighted Hardy spaces. By using more general weight vectors $\omega$ not confined to the unit weight sequence, it is still able to generate the K\"ahler information manifolds for signal processing filters. 
	
	Various well-known information manifolds are considered as special cases of the ($\phi$, $\omega$)-generalization of K\"ahlerian information geometry of a linear system. Obviously, when $\omega=(1,1,\cdots)$ is used, we obtain the unweighted complex cepstrum geometry that is the K\"ahlerian information geometry of a signal filter in the literature \cite{choi2015kahlerian}. These information manifolds in weighted Hardy spaces are not limited to K\"ahlerian information manifolds from the unweighted complex cepstrum \cite{choi2015kahlerian}, the geometry of the weighted stationarity filters, and mutual information geometry \cite{martin2000metric}.

\section{Example: ARMA and ARFIMA models}
\label{sec_whardy_example}
	In this section, we apply the theoretical framework developed in the previous section to time series models such as ARMA models and ARFIMA models in weighted Hardy spaces.

	The transfer function of the ARFIMA$(p,d,q)$ model with model parameters $\boldsymbol{\xi}=(\xi^{(-1)},\xi^{(0)},\cdots,\xi^{(p+q)})=(\sigma, d, \lambda_1,\cdots,\lambda_p, \mu_1, \cdots, \mu_q)$ is given by
	\begin{align}
		\label{tf_arfima}
		h(z;\boldsymbol{\xi})=\frac{\sigma^2}{2\pi}\frac{(1-\mu_{1}z^{-1})(1-\mu_{2}z^{-1})\cdots(1-\mu_{q}z^{-1})}{(1-\lambda_{1}z^{-1})(1-\lambda_{2}z^{-1})\cdots(1-\lambda_{p}z^{-1})}(1-z^{-1})^d,
	\end{align}
	where $\lambda_i$ ($|\lambda_i|<1$) is a pole from the AR part of the transfer function, and $\mu_i$ ($|\mu_i|<1$) is a zero from the MA part of the transfer function. It is straightforward that the ARMA transfer function is acquired by setting up $d=0$ in Eq. (\ref{tf_arfima}). 
	
	Among various smooth transformation functions of a transfer function, let us concentrate on the logarithmic function because of its extensive applications in signal processing and time series analysis. 
	
	Plugging Eq. (\ref{tf_arfima}) into the logarithmic function, the logarithmic transfer function of the ARFIMA system is obtained as
	\begin{align}
		\log{h(z;\boldsymbol{\xi})}=\sum_{i=1}^{p+q}\gamma_i\log{(1-\xi^{(i)} z^{-1})}+d\log{(1-z^{-1})}+\log{\frac{\sigma^2}{2\pi}},
	\end{align}
	where $\gamma_i=-1$ if $\xi^{(i)}$ ($1\le i \le p$) is a pole of the transfer function, and $\gamma_i=1$ if $\xi^{(i)}$ ($p+1\le i\le p+q$) is a root of the transfer function. It is obvious that $d=0$ implies the transfer function of ARMA models. 
	
	Considering the constant $\sigma$-submanifold, the $s$-th complex cepstrum coefficient $c_s$ of the ARFIMA model is derived from the definition of complex cepstrum in Eq. (\ref{cepstrum_coeff}):
	\begin{align}
	\label{cepstrum_coef_arfima}
	c_s=\left\{ 
	\begin{array}{ll}
	\frac{d+\sum_{i=1}^{p+q} \gamma_i (\xi^{(i)})^s}{s} & (s \neq 0)\\ 
	1 & (s=0),
	\end{array}
	\right.
\end{align}
	and it is noticeable that the complex cepstrum coefficients of the ARMA model are obtained by setting up $d=0$ in the cepstrum expression of Eq. (\ref{cepstrum_coef_arfima}).
	
\subsection{Arbitrary $\omega_s$}
	With the complex cepstrum of Eq. (\ref{cepstrum_coef_arfima}) and a weight vector $\omega$, the K\"ahler potential is easily found from Eq. (\ref{kahler_cepstrum}): 
	\begin{align}
	\mathcal{K}^{(\omega)}=\sum_{s=1}^{\infty} \omega_s\bigg|\frac{d+\sum_{i=1}^{p+q} \gamma_i (\xi^{(i)})^s}{s}\bigg|^2
+\omega_0.
	\label{kahler_weighted_arfima_general}
	\end{align}
	It is noteworthy that the convergence of the K\"ahler potential on the  K\"ahler-ARFIMA manifold is not always guaranteed and depends on the weight vector $\omega_s$ and filter parameters $\bold{\xi}$ such as the difference parameter, poles, and roots. However, since we consider only smooth transformation functions of signal filter transfer functions in weighted Hardy spaces, the convergence is obviously expected and we will comment otherwise. 
	
	Plugging the K\"ahler potential of Eq. (\ref{kahler_weighted_arfima_general}) into Eq. (\ref{metric_kahler_rel}), the metric tensor of the K\"ahler-ARFIMA manifolds is given by
	\begin{align}
	\label{eqn_metric_kahler_arfima_general}
		g_{u\bar{v}}=\left( 
			\begin{array}{cc}
				\sum_{s=1}^{\infty} \omega_s/s^2 & \gamma_j \sum_{s=1}^{\infty} \omega_s (\bar{\xi}^{(j)})^{s-1}/s \\ 
				\gamma_i \sum_{s=1}^{\infty} \omega_s (\xi^{(i)})^{s-1}/s  &\gamma_i\gamma_j \sum_{s=1}^{\infty} \omega_s  (\xi^{(i)}\bar{\xi}^{(j)})^{s-1}
			\end{array}
		\right),
	\end{align}
	where $u,v$ run from 0 to $p+q$, and $i,j$ run from 1 to $p+q$.
		
	Similar to the metric tensor, the Levi-Civita connection is also found from Eq. (\ref{connection_kahler_rel}) and Eq. (\ref{kahler_weighted_arfima_general}) as
	\begin{align}
		\label{eqn_connection_ijd_arfima_general}
		\Gamma_{ij, \bar{0}}&=-\gamma_i\sum_{s=1}^{\infty} \omega_s \bigg(1-\frac{1}{s}\bigg)(\xi^{(i)})^{s-2} \delta_{ij},\\
		\label{eqn_connection_ijk_arfima_general}
		\Gamma_{ij,\bar{k}}&=\gamma_j\gamma_k \sum_{s=1}^{\infty} \omega_s (s-1) (\xi^{(j)})^{s-2} (\bar{\xi}^{(k)})^{s-1}\delta_{ij},
	\end{align}
	where $i,j,k$ run from 1 to $p+q$. 
	
	Since the ARMA models are submodels of the ARFIMA model with $d=0$, the information manifolds of the ARMA models are also submanifolds of the ARFIMA models. Based on this fact, metric tensor components of the K\"ahler-ARMA geometry are obtained from those of the K\"ahler-ARFIMA geometry. From Eq. (\ref{eqn_metric_kahler_arfima_general}), the metric tensor of the K\"ahler-ARMA manifolds is found as
	\begin{align}
	\label{eqn_metric_kahler_arma_general}
		g_{i\bar{j}}=\gamma_i\gamma_j \sum_{s=1}^{\infty} \omega_s  (\xi^{(i)}\bar{\xi}^{(j)})^{s-1},
	\end{align}
	where $i,j$ run from 1 to $p+q$. It is also possible to derive the metric tensor directly from the K\"ahler potential by using Eq. (\ref{metric_kahler_rel}).
	
	The non-trivial components of Levi-Civita connection are derived not only from Eq. (\ref{eqn_connection_ijk_arfima_general}) but also from Eq. (\ref{connection_kahler_rel}):
	\begin{align}
		\label{eqn_connection_ijk_arma_general}
		\Gamma_{ij,\bar{k}}=\gamma_j\gamma_k \sum_{s=1}^{\infty} \omega_s (s-1) (\xi^{(j)})^{s-2} (\bar{\xi}^{(k)})^{s-1}\delta_{ij},
	\end{align}
	where $i,j,k$ run from 1 to $p+q$. 
	
	\subsection{$\omega_s=s^m$}
	For simplification, let us confine the weight sequence to $\omega_s=s^m$ for a real number $m$. As mentioned earlier, this weight sequence is related to the differentiation semi-norm spaces that are a building block for other weighted Hardy spaces. 
	
	According to Theorem \ref{thm_kahler_whardy} and Corollary \ref{crl_kahler_weighted_cepstrum}, geometry induced from finite weighted complex cepstrum norms is K\"ahler geometry. By using Theorem \ref{thm_kahler_potential_whardy} and Corollary \ref{crl_kahler_weighted_cepstrum}, not only applying the complex cepstrum coefficients of Eq. (\ref{cepstrum_coef_arfima}) to Eq. (\ref{kahler_cepstrum}) but also plugging $\omega_s=s^m$ to Eq. (\ref{kahler_weighted_arfima_general}) provide the K\"ahler potential of the K\"ahler-ARFIMA geometry given by
	\begin{align}
		\label{kahler_potential_weighted_arfima_general}
		\mathcal{K}^{(m)}=\sum_{s=1}^{\infty} s^{m}\bigg|\frac{d+\sum_{i=1}^{p+q} \gamma_i (\xi^{(i)})^s}{s}\bigg|^2
+\delta_{m,0},
	\end{align}
	where $\delta_{m,0}$ is the Kronecker delta.  
	
	It is noteworthy that the $|d|^2$-term in the K\"ahler potential can be divergent when $m\ge1$. One way of making finite complex cepstrum norms is to limit the range of $m$ to $m <1$. In this range of $m$, the weighted complex cepstrum norm of the ARFIMA model with $\omega_s=s^m$ becomes finite. When weighted complex cepstrum norms are convergent, the K\"ahler potential of the K\"ahler-ARFIMA geometry is decomposed into the following terms:
	\begin{align}
		\label{kahler_potential_weighted_cepstrum_arfima_general}
		\begin{split}
		\mathcal{K}^{(m)}=&\sum_{i,j=1}^{n}\gamma_i\gamma_j Li_{2-m}(\xi^{(i)}\bar{\xi}^{(j)})+\sum_{i=1}^{n} \gamma_i\big(d Li_{2-m}(\bar{\xi}^{(i)})+\bar{d} Li_{2-m}(\xi^{(i)})\big)\\
		&+|d|^2 Li_{2-m}(1)+\delta_{m,0},
		\end{split}
	\end{align}
	where $Li_t$ is the polylogarithm of order $t$.
	
	Another way of avoiding divergent complex cepstrum norms for arbitrary $m$ is regularizing the divergent term in the norm by setting $d=0$. Since the divergent term is dependent only on the difference parameter $d$, plugging $d=0$ not only removes the divergent term but also corresponds to the model reduction from the ARFIMA models to the ARMA models where the K\"ahler potential is always finite for any $m$ values. 
	
	By setting $d=0$, the K\"ahler potential of the ARMA geometry is given in the following form of
	\begin{align}
		\label{weighted_kahler_potential_arma}
		\mathcal{K}^{(m)}=\sum_{i,j=1}^{n}\gamma_i\gamma_j Li_{2-m}(\xi^{(i)}\bar{\xi}^{(j)})+\delta_{m,0},
	\end{align}
	where $Li_t$ is the polylogarithm of order $t$. It is obvious that the polylogarithm functions in the K\"ahler potential are all finite for the poles and the zeros of the ARMA transfer function in the unit disk $\mathbb{D}$.
	
	By using Eq. (\ref{metric_kahler_rel}), the metric tensor of the K\"ahler-ARFIMA geometry with the weight sequence of $\omega_s=s^m$ is obtained from the K\"ahler potential of Eq. (\ref{kahler_potential_weighted_cepstrum_arfima_general}). For $m<1$, the metric tensor of the K\"ahler-ARFIMA manifolds is given by
	\begin{align}
	\label{eqn_metric_kahler_arfima}
		g^{(m)}_{u\bar{v}}=\left( 
			\begin{array}{cc}
				Li_{2-m}(1) & \gamma_j Li_{1-m}(\bar{\xi}^{(j)})/\bar{\xi}^{(j)} \\ 
				\gamma_i Li_{1-m}(\xi^{(i)})/\xi^{(i)} &\gamma_i\gamma_jLi_{-m}(\xi^{(i)}\bar{\xi}^{(j)})/(\xi^{(i)}\bar{\xi}^{(j)})
\end{array}
\right),
\end{align}
	where $u,v$ run from 0 to $p+q$, and $i,j$ run from 1 to $p+q$. 
	
	In the weighted complex cepstrum geometry of the K\"ahler-ARFIMA model, the non-vanishing components of Levi-Civita connection are obtained by plugging from Eq. (\ref{kahler_potential_weighted_cepstrum_arfima_general}) to Eq. (\ref{connection_kahler_rel}):
	\begin{align}
		\label{eqn_connection_ijd_arfima}
		\Gamma_{ij, \bar{0}}&=-\gamma_j \frac{Li_{-m}(\xi^{(i)})-Li_{1-m}(\xi^{(i)})}{(\xi^{(i)})^2}\delta_{ij},\\
		\label{eqn_connection_ijk_arfima}
		\Gamma_{ij,\bar{k}}&=\gamma_j\gamma_k \frac{Li_{-m-1}(\xi^{(j)}\bar{\xi}^{(k)})-Li_{-m}(\xi^{(j)}\bar{\xi}^{(k)})}{(\xi^{(j)})^2\bar{\xi}^{(k)}}\delta_{ij},
	\end{align}
	where $i,j,k$ run from 1 to $p+q$. 
	
	Similar to the K\"ahler-ARFIMA geometry, the metric tensor of the K\"ahler-ARMA manifold is derived from partial derivatives on its K\"ahler potential of Eq. (\ref{weighted_kahler_potential_arma}). Opposite to the K\"ahler-ARFIMA geometry, the K\"ahler potential of the K\"ahler-ARMA geometry is finite for arbitrary $m$. Additionally, it is also possible to obtain the metric tensor from the submanifold of the K\"ahler-ARFIMA geometry, Eq. (\ref{eqn_metric_kahler_arfima}). The metric tensor of the K\"ahler-ARMA geometry is represented with
	\begin{align}
		g^{(m)}_{i\bar{j}}=\gamma_i \gamma_j Li_{-m}(\xi^{(i)}\bar{\xi}^{(j)})/(\xi^{(i)}\bar{\xi}^{(j)}),
	\end{align}
	 where $i,j$ run from 1 to $p+q$. It is straightforward to verify that the K\"ahler-ARMA manifolds are the submanifolds of the K\"ahler-ARFIMA manifolds. 
	 
	 In the K\"ahler-ARMA geometry, the non-trivial components of the Levi-Civita connection are represented with Eq. (\ref{eqn_connection_ijk_arfima}):
	 \begin{align}
		\label{eqn_connection_ijk_arma}
		\Gamma_{ij,\bar{k}}=\gamma_j\gamma_k \frac{Li_{-m-1}(\xi^{(j)}\bar{\xi}^{(k)})-Li_{-m}(\xi^{(j)}\bar{\xi}^{(k)})}{(\xi^{(j)})^2\bar{\xi}^{(k)}}\delta_{ij},
	\end{align}
	where $i,j,k$ run from 1 to $p+q$. The components also are derived from Eq. (\ref{connection_kahler_rel}), partial derivatives to the K\"ahler potential.
	 
\subsubsection{$m=0$}
	The weight vector for $m=0$ is the unit sequence of $\omega_s=1$ for all non-negative integers $s$. With the unit weight sequence, the weighted complex cepstrum norm is given by
	\begin{align}
		\mathcal{I}^{(0)}_\omega=\bigg(\sum_{s=0}^{\infty} |c_s|^2\bigg)^{1/2},
	\end{align}
	and this norm is exactly the unweighted complex cepstrum norm. Since the K\"ahler potential of the geometry with $m=0$ is the square of the unweighted complex cepstrum norm by Theorem \ref{thm_kahler_potential_whardy}, the K\"ahler potential of the geometry is represented with
	\begin{align}
		\label{kahler_potential_arfima_zero}
		\mathcal{K}^{(0)}=\sum_{s=1}^{\infty} \bigg|\frac{d+(\mu _{1}^{s}+\cdots +\mu_{q}^{s})-(\lambda _{1}^{s}+\cdots +\lambda _{p}^{s})}{s}\bigg|^2+1.
	\end{align}
	It is straightforward to check that the K\"ahler potential of Eq. (\ref{kahler_potential_arfima_zero}) is identical up to the last constant term on the right-hand side to the K\"ahler potential of the K\"ahler-ARFIMA geometry in Choi and Mullhaupt \cite{choi2015geometric} where the ARFIMA model in the literature was scaled up to the constant gain. Considering the $\sigma$-submanifolds, it is straightforward to show that the K\"ahler potential of Eq. (\ref{kahler_potential_arfima_zero}) is identical to the K\"ahler potential in the literature. When $d=0$ in Eq. (\ref{kahler_potential_arfima_zero}), the K\"ahler potential is same up to the gain term with the K\"ahler potential of the ARMA geometry in Choi and Mullhaupt \cite{choi2015kahlerian} where the model is also $\sigma$-constant.
	
	 The metric tensor of the geometry is obtained by either plugging $m=0$ to Eq. (\ref{eqn_metric_kahler_arfima}) or taking partial derivatives to the K\"ahler potential of Eq. (\ref{kahler_potential_arfima_zero}). For $m=0$, the metric tensor of the K\"ahler-ARFIMA geometry is found as
\begin{align}
	g^{(0)}_{u\bar{v}}=\left( 
		\begin{array}{ccc}
		\frac{\pi^2}{6}  & \frac{1}{\bar{\lambda}_{j}}\log {(1-\bar{\lambda}_{j})}& -\frac{1}{\bar{\mu}_{j}}\log {(1-\bar{\mu}_{j})} \\ 
		\frac{1}{\lambda _{i}}\log {(1-\lambda _{i})} & \frac{1}{1-\lambda_i\bar{\lambda}_j} & -\frac{1}{1-\lambda _{i}\bar{\mu}_{j}} \\ 
		-\frac{1}{\mu _{i}}\log {(1-\mu _{i})} & -\frac{1}{1-\mu _{i}\bar{\lambda}_{j}} & \frac{1}{1-\mu _{i}\bar{\mu}_{j}}
		\end{array}
	\right),
	\label{eqn_metric_kahler_arfima_zero}
\end{align}
	where $u,v$ run for $(\xi^{(0)}=d,\xi^{(1)}=\lambda_1, \cdots,\xi^{(p)}=\lambda_p, \xi^{(p+1)}=\mu_1,\cdots,\xi^{(p+q)}=\mu_q)$, the index for $\lambda_i$ runs from 1 to $p$, and the index for $\mu_i$ runs from 1 to $q$. As described above, the metric tensor is matched to the metric tensor for the K\"ahler information geometry of the ARFIMA models \cite{choi2015geometric}.
	
	The non-vanishing components of the Levi-Civita connection on the K\"ahler-ARFIMA manifolds are also derived not only from Eq. (\ref{eqn_connection_ijd_arfima}) and Eq. (\ref{eqn_connection_ijk_arfima}) but also the K\"ahler potential of Eq. (\ref{kahler_potential_arfima_zero}):
	\begin{align}
		\label{eqn_connection_ijd_arfima_zero}
		\Gamma_{ij, \bar{0}}&=-\gamma_j \frac{\frac{\xi^{(i)}}{1-\xi^{(i)}}+\log{(1-\xi^{(i)})}}{(\xi^{(i)})^2}\delta_{ij},\\
		\Gamma_{ij,\bar{k}}&=\gamma_j\gamma_k \frac{\bar{\xi}^{(k)}}{(1-\xi^{(j)}\bar{\xi}^{(k)})^2}\delta_{ij},
	\label{eqn_connection_ijk_arfima_zero}
	\end{align}
	where $i,j,k$ run from 1 to $p+q$ and $\delta_{ij}$ is the Kronecker delta. These connection components were not calculated previously in the K\"ahler-ARFIMA geometry \cite{choi2015geometric}. However, $\Gamma_{ij,\bar{k}}$ components in this K\"ahler-ARFIMA geometry are also matched with those in the K\"ahler information geometry of the ARMA models \cite{choi2015kahlerian}. 
	
	Similar to the K\"ahler-ARFIMA geometry, the $m=0$ geometry of the K\"ahler-ARMA models is derived from the identical procedure. The metric tensor of the K\"ahler-ARMA manifolds is given by
	\begin{align}
	g^{(0)}_{u\bar{v}}=\left( 
		\begin{array}{cc}
		\frac{1}{1-\lambda_i\bar{\lambda}_j} & -\frac{1}{1-\lambda _{i}\bar{\mu}_{j}} \\ 
		-\frac{1}{1-\mu _{i}\bar{\lambda}_{j}} & \frac{1}{1-\mu _{i}\bar{\mu}_{j}}
		\end{array}
	\right),
	\label{eqn_metric_kahler_arma_zero}
\end{align}
	and this metric tensor is also consistent with the metric tensor for the K\"ahlerian information geometry of the ARMA models \cite{choi2015kahlerian}. 
	
	The non-vanishing Levi-Civita connection components in the K\"ahler-ARMA geometry are $\Gamma_{ij,\bar{k}}$ in the K\"ahler-ARFIMA geometry:
	\begin{align}
		\Gamma_{ij,\bar{k}}=\gamma_j\gamma_k \frac{\bar{\xi}^{(k)}}{(1-\xi^{(j)}\bar{\xi}^{(k)})^2}\delta_{ij},
		\label{eqn_connection_ijk_arma_zero}
	\end{align} 
	where $i,j,k$ run from 1 to $p+q$ and $\delta_{ij}$ is the Kronecker delta. Since the geometry in the literature \cite{choi2015kahlerian} is the constant-gain submanifold of the full K\"ahler-ARMA manifold, the non-trivial components of the Levi-Civita connection on the $\sigma$-submanifold are only $\Gamma_{ij,\bar{k}}$, identical to $\Gamma_{ij,\bar{k}}$ of the full K\"ahler-ARMA geometry in this paper.

\subsubsection{$m=1$}	
	When $m=1$, the weighted complex cepstrum norm with the weight sequence $\omega_s=s$ for all positive integers $s$ is given by
	\begin{align}
		\mathcal{I}^{(1)}_\omega=\bigg(\sum_{s=1}^{\infty} s|c_s|^2\bigg)^{1/2},
	\end{align}
	and it is also known as the Hilbert-Schmidt norm of the Hankel matrix in complex cepstrum coefficients. Additionally, the norm is related to mutual information between past and future in ARMA model studied by \cite{oppenheim1965superposition, martin2000metric}. 
	
	Since the K\"ahler potential of the geometry is the square of the weighted complex cepstrum norm, the K\"ahler potential of the K\"ahler-ARFIMA geometry is provided by Eq. (\ref{kahler_potential_weighted_arfima_general}):
	\begin{align}
		\mathcal{K}^{(1)}=\sum_{s=1}^{\infty} \frac{|d+(\mu _{1}^{s}+\cdots +\mu_{q}^{s})-(\lambda _{1}^{s}+\cdots +\lambda _{p}^{s})|^2}{s}.
	\end{align}
	It is noteworthy that there is no $\sigma$-dependence in the expression of the K\"ahler potential. However, this K\"ahler potential of the K\"ahler-ARFIMA models with the weight sequence of $\omega_s=s$ is not guaranteed as finite because $\sum_{s=1}^{\infty}\frac{|d|^2}{s}$ is divergent as mentioned above. 
	
	With $m=1$, the K\"ahler potential of the ARFIMA geometry from the weighted complex cepstrum norm is convergent if and only if $d=0$, i.e., the ARFIMA model is reduced to the ARMA model. In this case, the K\"ahler potential of the K\"ahler-ARMA geometry from the weighted complex cepstrum norms is found from Eq. (\ref{weighted_kahler_potential_arma}):
	\begin{align}
		\mathcal{K}^{(1)}=\sum_{i,j=1}^{n} \gamma_i\gamma_jLi_{1}(\xi^i\bar{\xi}^j)=\log{\bigg(\frac{\prod_{i=1}^{p}\prod_{j=1}^{q}(1-\lambda_i\bar{\mu}_j)\prod_{i=1}^{p}\prod_{j=1}^{q}(1-\mu_j\bar{\lambda}_i)}{\prod_{i=1}^{p}\prod_{j=1}^{p}(1-\lambda_i\bar{\lambda}_j)\prod_{i=1}^{q}\prod_{j=1}^{q}(1-\mu_i\bar{\mu}_j)}\bigg)},
	\end{align}
	where $Li_1(z)=-\log{(1-z)}$. Considering that the K\"ahler potential is the square of the weighted complex cepstrum norm by Corollary \ref{crl_kahler_weighted_cepstrum}, the K\"ahler potential of the ARMA model for $m=1$  is the mutual information between past and future in the ARMA processes studied by Martin \cite{martin2000metric}. It is also interesting that the parameters in the weighted complex cepstrum norm are only the roots and the poles of the transfer function in the ARMA model. The fact that the length measure is independent with the gain parameter $\sigma$ is also consistent with the $\sigma$-quotient model mentioned in Martin \cite{martin2000metric}. 
	
	From the K\"ahler potential found above, the metric tensor of the K\"ahler-ARMA geometry with $m=1$ is given by
\begin{align*}
	g^{(1)}_{u\bar{v}}=\left( 
	\begin{array}{ccc}
	\frac{1}{(1-\lambda _{i}\bar{\lambda}_{j})^2} & -\frac{1}{(1-\lambda_{i}\bar{\mu}_{j})^2} \\ 
	-\frac{1}{(1-\mu _{i}\bar{\lambda}_{j})^2} & \frac{1}{(1-\mu_{i}\bar{\mu}_{j})^2}
	\end{array}\right),
\end{align*}
	where $u,v$ run for $(\xi^{(1)}=\lambda_1, \cdots,\xi^{(p)}=\lambda_p, \xi^{(p+1)}=\mu_1,\cdots,\xi^{(p+q)}=\mu_q)$, the index for $\lambda_i$ runs from 1 to $p$, and the index for $\mu_i$ runs from 1 to $q$. The metric tensor is identical to the metric tensor from mutual information between past and future of ARMA models in Martin \cite{martin2000metric}. If only the pure AR part or the pure MA part of the metric tensor are considered, the metric tensor of the submanifold is considered as the Poincar\'e metric of polydisc. Additionally, this metric is the Bergman metric defined above.
	
	It is also straightforward to calculate non-vanishing components of the Levi-Civita connection on the K\"ahler-ARMA manifolds induced from the mutual information between past and future in the ARMA model. By using Eq. (\ref{eqn_connection_ijk_arfima_general}), the connection components are given by
	\begin{align}
		\Gamma_{ij,\bar{k}}=\gamma_j\gamma_k \frac{2\bar{\xi}^{(k)}}{(1-\xi^{(j)}\bar{\xi}^{(k)})^3}\delta_{ij},
	\end{align}
	and all other components are vanishing.

	\subsection{$\omega_s=|\rho|^{2s}$}
	By plugging $\omega_s=|\rho|^{2s}$ into Eq. (\ref{kahler_weighted_arfima_general}), the K\"ahler potential of the K\"ahler-ARFIMA manifold with the exponentiation factor is represented with
	\begin{align}
		\mathcal{K}=\sum_{s=1}^{\infty} |\rho|^{2s}\bigg|\frac{d+\sum_{i=1}^{p+q} \gamma_i (\xi^{(i)})^s}{s}\bigg|^2+1,
		\label{kahler_weighted_arfima_expw}
	\end{align}
	and it is easy to check that in the limit of $|\rho| \to 1$, Eq. (\ref{kahler_weighted_arfima_expw}) is converged to the unweighted case, Eq. (\ref{kahler_potential_arfima_zero}).
	
	By either plugging the weight vector to Eq. (\ref{eqn_metric_kahler_arfima_general}) or taking partial derivatives of Eq. (\ref{kahler_weighted_arfima_expw}), the metric tensor of the K\"ahler-ARFIMA manifold with the exponentiation is given in the following form of
	\begin{align}
	\label{eqn_metric_kahler_arfima_expw}
		g_{u\bar{v}}=\left( 
			\begin{array}{cc}
				\sum_{s=1}^{\infty} |\rho|^{2s}/s^2 & \gamma_j \sum_{s=1}^{\infty} |\rho|^{2s} (\bar{\xi}^{(j)})^{s-1}/s \\ 
				\gamma_i \sum_{s=1}^{\infty} |\rho|^{2s} (\xi^{(i)})^{s-1}/s  &\gamma_i\gamma_j \sum_{s=1}^{\infty} |\rho|^{2s}  (\xi^{(i)}\bar{\xi}^{(j)})^{s-1}
			\end{array}
		\right),
	\end{align}
	and it is also straightforward to check that Eq. (\ref{eqn_metric_kahler_arfima_expw}) is reduced to the unweighted case, Eq. (\ref{eqn_metric_kahler_arfima_zero}) as $|\rho| \to 1$.
	
	Similarly, the Levi-Civita connection of the K\"ahler-ARFIMA manifold with the exponentiation factor is also found as
	\begin{align}
		\label{eqn_connection_ijd_arfima_expw}
		\Gamma_{ij, \bar{0}}&=-\gamma_i\sum_{s=1}^{\infty} |\rho|^{2s} \bigg(1-\frac{1}{s}\bigg)(\xi^{(i)})^{s-2} \delta_{ij},\\
		\label{eqn_connection_ijk_arfima_expw}
		\Gamma_{ij,\bar{k}}&=\gamma_j\gamma_k \sum_{s=1}^{\infty} |\rho|^{2s} (s-1) (\xi^{(j)})^{s-2} (\bar{\xi}^{(k)})^{s-1}\delta_{ij},
	\end{align}
	where $i,j,k$ run from 1 to $p+q$. Both Eq. (\ref{eqn_connection_ijd_arfima_expw}) and Eq. (\ref{eqn_connection_ijk_arfima_expw}) converge to Eq. (\ref{eqn_connection_ijd_arfima_zero}) and Eq. (\ref{eqn_connection_ijk_arfima_zero}) in the limit of $|\rho| \to 1$, respectively.
	
	As mentioned above, the ARMA model is a special case of the ARFIMA model with $d=0$. The information manifolds of the ARMA models are indeed submanifolds of the ARFIMA geometry. Based on this fact, the metric tensor components of the K\"ahler-ARMA manifolds are easily acquired from those of the K\"ahler-ARFIMA manifolds. By taking the components of the metric tensor from Eq. (\ref{eqn_metric_kahler_arfima_expw}), the metric tensor of the K\"ahler-ARMA manifolds is found as
	\begin{align}
	\label{eqn_metric_kahler_arma_expw}
		g_{i\bar{j}}=\gamma_i\gamma_j \sum_{s=1}^{\infty} |\rho|^{2s}  (\xi^{(i)}\bar{\xi}^{(j)})^{s-1},
	\end{align}
	where $i,j$ run from 1 to $p+q$. It is also possible to gain the metric tensor directly from the K\"ahler potential by using Eq. (\ref{metric_kahler_rel}). Similar to the ARFIMA case, Eq. (\ref{eqn_metric_kahler_arma_expw}) is converging to the unweighted case, Eq. (\ref{eqn_metric_kahler_arma_zero}) as $|\rho| \to 1$.
	
	The non-trivial Levi-Civita connection components of the K\"ahler-ARMA geometry are derived not only from Eq. (\ref{eqn_connection_ijk_arfima_general}) but also from Eq. (\ref{connection_kahler_rel}). Additionally, those can be obtained from Eq. (\ref{eqn_connection_ijk_arfima_expw}):
	\begin{align}
		\label{eqn_connection_ijk_arma_expw}
		\Gamma_{ij,\bar{k}}=\gamma_j\gamma_k \sum_{s=1}^{\infty} |\rho|^{2s} (s-1) (\xi^{(j)})^{s-2} (\bar{\xi}^{(k)})^{s-1}\delta_{ij},
	\end{align}
	where $i,j,k$ run from 1 to $p+q$. By taking $|\rho| \to 1$, the Levi-Civita connection also converge to the unweighted case, Eq. (\ref{eqn_connection_ijk_arma_zero}).
	
\section{Conclusion}
\label{sec_whardy_conclusion}
	In this paper, we investigate the geometric properties of linear systems within the framework of weighted Hardy spaces. We show that the information manifold of signal filters, induced by the weighted Hardy norms of composite functions formed by smooth transformations and transfer functions, is a K\"ahler manifold. The weighted Hardy norm of the transformed transfer function is directly related to the K\"ahler potential of the corresponding signal filter geometry. With the structural properties of K\"ahler manifolds, key geometric objects, such as the metric tensor, the Levi-Civita connection, and the Ricci tensor, can be efficiently derived from the K\"ahler potential.
	
	The use of weighted Hardy spaces with weight sequences $\omega$, along with smooth transformations $\phi$ of transfer functions, constitutes what we refer to as the ($\phi$,$\omega$)-generalization of the K\"ahler information geometry developed in the context of the unweighted Hardy space \cite{choi2015kahlerian}. This generalized framework provides a unified and flexible approach for identifying and analyzing the underlying information manifolds associated with signal processing filters across various weighted Hardy spaces.
	
	As a special case of weighted Hardy norms, the weighted complex cepstrum norms of a linear system also give rise to information manifolds equipped with K\"ahler structures. In this setting, the K\"ahler potential of the complex cepstrum geometry is given by the square of the weighted complex cepstrum norm. The unweighted case corresponds to the K\"ahlerian information geometry developed in Choi and Mullhaupt \cite{choi2015kahlerian}. Moreover, for weight sequences of the form $\omega_s = |\rho|^{2s}$ , the weighted Hardy norms of signal filters are equivalent to the unweighted Hardy norms of signal filters with the exponentiation of $\rho$, i.e., information geometry of signal filters with the exponentiation of $\rho$ can be derived from the concept of weighted Hardy norms.

	As applications of the K\"ahler geometric approach based on weighted Hardy norms, we exemplify ARMA and ARFIMA models with arbitrary weight vectors. Several interesting information manifolds naturally emerge from this framework. As special cases, these include the K\"ahlerian information manifolds associated with the unweighted complex cepstrum norm \cite{choi2015kahlerian}, the geometric structure related to the mutual information between past and future \cite{martin2000metric}, and the information geometry of signal filters involving exponentiation. All of these cases can be understood within a unified framework based on K\"ahler manifolds and weighted Hardy spaces.

\section*{Acknowledgment}
	We thank Xiang Shi and Andrew P. Mullhaupt for discussions on mutual information between past and future in ARMA models \cite{martin2000metric}.

\bibliographystyle{plain}
\bibliography{KahlerInfoWeightedHardy}

\end{document}